\documentclass{ws-procs975x65}

\begin{document}

\title{Modelling the neutrino in terms of Cosserat elasticity }

\author{O. CHERVOVA$^*$ and D. VASSILIEV$^{**}$}

\address{Department of Mathematics and Institute of Origins,\\
University College London,\\
Gower Street, London WC1E~6BT, UK\\
$^*$E-mail: olgac@math.ucl.ac.uk\\
$^{**}$E-mail: D.Vassiliev@ucl.ac.uk}

\begin{abstract}
The paper deals with the Weyl equation which is the massless Dirac
equation. We study the Weyl equation in the stationary setting, i.e.
when the the spinor field oscillates harmonically in time. We
suggest a new geometric interpretation of the stationary Weyl
equation, one which does not require the use of spinors, Pauli
matrices or covariant differentiation. We think of our 3-dimensional
space as an elastic continuum and assume that material points of
this continuum can experience no displacements, only rotations. This
framework is a special case of the Cosserat theory of elasticity.
Rotations of material points of the space continuum are described
mathematically by attaching to each geometric point an orthonormal
basis which gives a field of orthonormal bases called the coframe.
As the dynamical variables (unknowns) of our theory we choose the
coframe and a density. We choose a particular potential energy which
is conformally invariant and then incorporate time into our
action in the standard Newtonian way, by subtracting kinetic energy.
The main result of our paper is the theorem stating that in the
stationary setting our model is equivalent to a pair of Weyl
equations. The crucial element of the proof is the observation that
our Lagrangian admits a factorisation.
\end{abstract}


\bodymatter

\section{Our model}
\label{Our model}

In this paper we work on a 3-manifold $M$ equipped with
local coordinates $x^\alpha$, $\alpha=1,2,3$, and prescribed
positive metric $g_{\alpha\beta}$ which does not depend on time $x^0$. We
view our 3-manifold $M$ as an elastic continuum whose material
points can experience no displacements, only rotations, with
rotations of different material points being independent.
Rotations of material points of the elastic continuum
are described mathematically by attaching to each geometric point of
the manifold $M$ an orthonormal basis, which gives a field of
orthonormal bases called the \emph{coframe}.

The coframe $\vartheta$ is a triple of orthonormal covector fields
$\vartheta^j$, $j=1,2,3$, on the 3-manifold $M$. Each covector field
$\vartheta^j$ can be written more explicitly as
$\vartheta^j{}_\alpha$ where the tensor index $\alpha=1,2,3$
enumerates the components. The orthonormality condition for the
coframe can be represented as a single tensor identity
\begin{equation}
\label{constraint for coframe}
g=\delta_{jk}\vartheta^j\otimes\vartheta^k.
\end{equation}
We view the identity (\ref{constraint for coframe}) as a kinematic
constraint: the metric $g$ is given (prescribed) and the coframe
elements $\vartheta^j$ are chosen so that they satisfy
(\ref{constraint for coframe}), which leaves us with three real
degrees of freedom at every point of $M$.

As dynamical variables in our model we choose the coframe
$\vartheta$ and a positive density $\rho$. These
are functions of local coordinates $x^\alpha$ on $M$ as well as of
time $x^0$.

At a physical level, making the density $\rho$ a dynamical
variable means that we view our continuum more like a
fluid rather than a solid. In other words, we allow the
material to redistribute itself so that it finds its equilibrium
distribution.

The crucial element in our construction is the choice of potential
energy. As a measure of deformations caused by rotations of material
points we choose \emph{axial torsion}, which is the 3-form
given by the explicit formula
$T^\mathrm{ax}:=
\frac13\delta_{jk}\vartheta^j\wedge d\vartheta^k$
where $\,d\,$ denotes the exterior derivative.
We take the potential energy of our continuum to be
$P(x^0):=\int_M\|T^\mathrm{ax}\|^2\rho\,dx^1dx^2dx^3$.
It is easy to see that this potential energy
is conformally invariant: it does not change if we rescale our coframe as
$\vartheta^j\mapsto e^h\vartheta^j$
and our density as
$\rho\mapsto e^{2h}\rho$
where $h:M\to\mathbb{R}$ is an arbitrary scalar function.

We take the kinetic energy of our continuum to be
$K(x^0):=\int_M\|\dot\vartheta\|^2\rho\,dx^1dx^2dx^3$
where $\dot\vartheta$ is the 2-form
$\dot\vartheta:=
\frac13\delta_{jk}\vartheta^j\wedge\partial_0\vartheta^k$
with $\partial_0$ denoting the time derivative.
The 2-form $\dot\vartheta$ is, up to a constant
factor, the Hodge dual of the vector of angular velocity.

We now combine the potential energy and
kinetic energy to form the action
(variational functional) of our dynamic problem:
\begin{equation}
\label{our action}
S(\vartheta,\rho):=\int_{\mathbb{R}}(P(x^0)-K(x^0))\,dx^0
=\int_{{\mathbb{R}}\times M}L(\vartheta,\rho)\,dx^0dx^1dx^2dx^3
\end{equation}
where
\begin{equation}
\label{our Lagrangian density}
L(\vartheta,\rho):=(\|T^\mathrm{ax}\|^2-\|\dot\vartheta\|^2)\rho
\end{equation}
is our dynamic (time-dependent) Lagrangian density.

Our field equations (Euler--Lagrange equations) are obtained by
varying the action (\ref{our action}) with respect to the
coframe $\vartheta$ and density $\rho$. Varying with respect to the
density $\rho$ is easy: this gives the field equation
$\|T^\mathrm{ax}\|^2=\|\dot\vartheta\|^2$ which is equivalent to
$L(\vartheta,\rho)=0$. Varying with respect to the coframe
$\vartheta$ is more difficult because we have to maintain the kinematic
constraint (\ref{constraint for coframe}).

\section{Switching to the language of spinors}
\label{Switching to the language of spinors}

The technical difficulty mentioned above can be overcome by switching
to a different dynamical variable. It is known that
in dimension~$3$
a coframe $\vartheta$ and a (positive) density $\rho$ are equivalent
to a nonvanishing spinor field $\xi$ modulo the sign of $\xi$.
The explicit formulae relating coframes and spinors are given
in Appendix C of Ref.~\citen{prd2010}.
The advantage of switching to a spinor field $\xi$ is that there are no
kinematic constraints on its components, so the derivation
of field equations becomes straightforward.

Switching to spinors in formula
(\ref{our Lagrangian density})
we arrive at the following self-contained explicit spinor representation
of our dynamic Lagrangian density
\begin{multline}
\label{our Lagrangian density in terms of spinor}
L(\xi)=
\frac4{9\bar\xi^{\dot c}\sigma_{0\dot cd}\xi^d}
\Bigl(
[i(
\bar\xi^{\dot a}\sigma^\alpha{}_{\dot ab}\nabla_\alpha\xi^b
-
\xi^b\sigma^\alpha{}_{\dot ab}\nabla_\alpha\bar\xi^{\dot a}
)]^2
\\
-
\|i(
\bar\xi^{\dot a}\sigma_{\alpha\dot ab}\partial_0\xi^b
-
\xi^b\sigma_{\alpha\dot ab}\partial_0\bar\xi^{\dot a}
)\|^2
\Bigr)\sqrt{\operatorname{det}g}
\end{multline}
where the $\sigma$ are Pauli matrices and the $\nabla$ are covariant derivatives.

\section{Separating out time}
\label{Separating out time}

We write down the dynamic
(containing time derivatives)
field equation for the Lagrangian density
(\ref{our Lagrangian density in terms of spinor}) and seek
solutions of the form
\begin{equation}
\label{stationary spinor field}
\xi(x^0,x^1,x^2,x^3)=
e^{-ip_0x^0}\eta(x^1,x^2,x^3)
\end{equation}
where $p_0\ne0$
is a real number. We call solutions of the form
(\ref{stationary spinor field}) \emph{stationary}.

It turns out that despite the fact that our dynamic field equation
is nonlinear, time $x^0$ can be separated out and
stationary solutions do indeed make sense.
The underlying group-theoretic reason
for our nonlinear dynamic field equation
admitting separation
of variables is the fact that our model is $\mathrm{U}(1)$-invariant,
i.e.~it is invariant under the multiplication of a spinor field
by a complex constant of modulus 1.

Our problem has been reduced to the study of
the stationary (time-independent) Lagrangian density
\begin{equation}
\label{our Lagrangian density in terms of spinor stationary}
L(\eta)=
\frac{16}{9\bar\eta^{\dot c}\sigma_{0\dot cd}\eta^d}
\Bigl(
\Bigl[\frac i2(
\bar\eta^{\dot a}\sigma^\alpha{}_{\dot ab}\nabla_\alpha\eta^b
-
\eta^b\sigma^\alpha{}_{\dot ab}\nabla_\alpha\bar\eta^{\dot a}
)\Bigr]^2
-
(p_0\bar\eta^{\dot a}\sigma_{0\dot ab}\eta^b)^2
\Bigr)\sqrt{\operatorname{det}g}
\end{equation}
which is our dynamic Lagrangian
density (\ref{our Lagrangian density in terms of spinor})
with time $x^0$ separated out.

\section{Main result}
\label{Main result}

Our main result is the following

\begin{theorem}
\label{main theorem}
A nonvanishing time-independent spinor field $\eta$
is a solution of the field equation
for our stationary Lagrangian
density~(\ref{our Lagrangian density in terms of spinor stationary})
if and only if it is a solution of one of
the two stationary Weyl equations
\begin{equation}
\label{stationary Weyl equation}
\pm p_0\sigma^0{}_{\dot ab}\eta^b
+i\sigma^\alpha{}_{\dot ab}\nabla_\alpha\eta^b=0.
\end{equation}
\end{theorem}

\begin{proof}
Observe that our stationary Lagrangian
density~(\ref{our Lagrangian density in terms of spinor stationary})
factorises as
\begin{equation}
\label{factorization formula}
L(\eta)=-\frac{32p_0}9
\frac{L_+(\eta)\,L_-(\eta)}{L_+(\eta)-L_-(\eta)}
\end{equation}
where
$
L_{\pm}(\eta):=
\Bigl[
\frac i2(
\bar\eta^{\dot a}\sigma^\alpha{}_{\dot ab}\nabla_\alpha\eta^b
-
\eta^b\sigma^\alpha{}_{\dot ab}\nabla_\alpha\bar\eta^{\dot a}
)
\pm p_0\bar\eta^{\dot a}\sigma^0{}_{\dot ab}\eta^b
\Bigr]\sqrt{\operatorname{det}g}
$
are the Lagrangian densities for the stationary Weyl equations
(\ref{stationary Weyl equation}).
It is easy to see that the latter posses the property of scaling covariance:
\begin{equation}
\label{proof of theorem equation 1}
L_\pm(e^h\eta)=e^{2h}L_\pm(\eta)
\end{equation}
where $h:M\to\mathbb{R}$ is an arbitrary scalar function.
In
fact, the Lagrangian density of any formally selfadjoint
(symmetric) linear first order partial differential operator
has the scaling covariance property (\ref{proof of theorem equation 1}).

The abstract argument presented in Section 6 of Ref.~\citen{prd2010}
shows that properties
(\ref{factorization formula}) and (\ref{proof of theorem equation 1})
imply the statement of Theorem \ref{main theorem}.
\end{proof}

\end{document}